\documentclass[a4paper,UKenglish]{lipics-v2016}
\usepackage{microtype}
\usepackage{amssymb}
\usepackage{thm-restate}
\usepackage{hyperref}
\hypersetup{colorlinks,linkcolor=blue,citecolor=blue,urlcolor=blue}
\usepackage{tikz}
\usetikzlibrary{arrows,calc,automata}
\tikzset{LMC style/.style={>=angle 60,every edge/.append style={thick},every state/.style={thick,minimum size=20,inner sep=0.5}}}

\usepackage{ifthen}
\newcommand{\techRep}{false} 
\newcommand{\iftechrep}{\ifthenelse{\equal{\techRep}{false}}}

\theoremstyle{plain}
\newtheorem{proposition}[theorem]{Proposition}

\theoremstyle{definition}
\newtheorem{remarkS}[theorem]{Remark}

\newcommand{\acc}{f_+}
\newcommand{\eow}{\$}
\newcommand{\A}{\mathcal{A}}
\newcommand{\F}{\mathbb{F}}
\newcommand{\M}{\mathcal{M}}
\newcommand{\N}{\mathbb{N}}
\newcommand{\nacc}{\#\mathit{acc}}
\newcommand{\pmin}{p_\mathit{min}}
\newcommand{\Q}{\mathbb{Q}}
\newcommand{\R}{\mathbb{R}}
\newcommand{\round}[1]{\langle #1 \rangle}
\newcommand{\rej}{f_-}
\newcommand{\T}{\mathcal{T}}
\renewcommand{\vec}[1]{\mathbf{#1}}
\newcommand{\Z}{\mathbb{Z}}

\title{On Computing the Total Variation Distance of Hidden Markov Models}
\author{Stefan Kiefer}
\affil{University of Oxford, United Kingdom}

\authorrunning{S. Kiefer} 

\Copyright{Stefan Kiefer}

\subjclass{F.1.1 Models of Computation, F.2.1 Numerical Algorithms and Problems, G.3 Probability and Statistics}
\keywords{Labelled Markov Chains, Hidden Markov Models, Distance, Decidability, Complexity}

\EventEditors{Ioannis Chatzigiannakis, Christos Kaklamanis, Daniel Marx, and Don Sannella}
\EventNoEds{4}
\EventLongTitle{45th International Colloquium on Automata, Languages, and Programming (ICALP 2018)}
\EventShortTitle{ICALP 2018}
\EventAcronym{ICALP}
\EventYear{2018}
\EventDate{July 9--13, 2018}
\EventLocation{Prague, Czech Republic}
\EventLogo{eatcs}
\SeriesVolume{80}
\ArticleNo{}

\begin{document}

\pagestyle{headings}  

\maketitle              

\begin{abstract}
We prove results on the decidability and complexity of computing the total variation distance (equivalently, the $L_1$-distance) of hidden Markov models (equivalently, labelled Markov chains).
This distance measures the difference between the distributions on words that two hidden Markov models induce.
The main results are: (1) it is undecidable whether the distance is greater than a given threshold; (2)~approximation is \#P-hard and in PSPACE.
\end{abstract}

\section{Introduction} \label{sec-intro}

A (discrete-time, finite-state, finite-word) \emph{labelled Markov chain (LMC)} (often called \emph{hidden Markov model}) has a finite set $Q$ of states and for each state a probability distribution over its outgoing transitions.
Each outgoing transition is labelled with a letter from an alphabet~$\Sigma$ and leads to a target state, or is labelled with an end-of-word symbol~$\eow$.
Here are two LMCs:
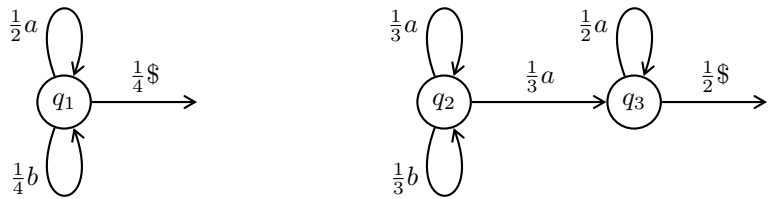
\begin{center}
\begin{tikzpicture}[scale=2.5,LMC style]
\node[state] (q1) at (0,0) {$q_1$};
\path[->] (q1) edge [loop,out=110,in=70,looseness=13] node[pos=0.25,left] {$\frac12 a$} (q1);
\path[->] (q1) edge [loop,out=250,in=290,looseness=13] node[pos=0.25,left] {$\frac14 b$} (q1);
\path[->] (q1) edge node[above] {$\frac14 \eow$} (0.7,0);
\node[state] (q2) at (2,0) {$q_2$};
\node[state] (q3) at (3,0) {$q_3$};
\path[->] (q2) edge [loop,out=110,in=70,looseness=13] node[pos=0.25,left] {$\frac13 a$} (q2);
\path[->] (q2) edge [loop,out=250,in=290,looseness=13] node[pos=0.25,left] {$\frac13 b$} (q2);
\path[->] (q2) edge node[above] {$\frac13 a$} (q3);
\path[->] (q3) edge [loop,out=110,in=70,looseness=13] node[pos=0.25,left] {$\frac12 a$} (q3);
\path[->] (q3) edge node[above] {$\frac12 \eow$} (3.7,0);
\end{tikzpicture}
\end{center}
The LMC starts in a given initial state (or in a random state according to a given initial distribution), picks a random transition according to the state's distribution over the outgoing transitions, outputs the transition label, moves to the target state, and repeats until the end-of-word label~$\eow$ is emitted. 
This induces a probability distribution over finite words (excluding the end-of-word label~$\eow$).
In the example above, if $q_1$ and $q_2$ are the initial states then the LMCs induce distributions $\pi_1, \pi_2$ with $\pi_1(a a) = \frac12 \cdot \frac12 \cdot \frac14$ and $\pi_2(a a) = \frac13 \cdot \frac13 \cdot \frac12 + \frac13 \cdot \frac12 \cdot \frac12$.
LMCs are widely employed in fields such as speech recognition (see~\cite{Rabiner89} for a tutorial),
gesture recognition~\cite{Gesture},
signal processing~\cite{SignalProcessing},
and climate modeling~\cite{Weather}.
LMCs are heavily used in computational biology~\cite{HMM-comp-biology},
more specifically in DNA modeling~\cite{DNA-modeling} and biological sequence analysis~\cite{durbin1998biological},
including protein structure prediction~\cite{ProteinStructure} 
and gene finding~\cite{GeneFinding}.
In computer-aided verification, LMCs are the most fundamental model for probabilistic systems; model-checking tools such as Prism~\cite{KNP11} or Storm~\cite{Storm} are based on analyzing LMCs efficiently.
 
A fundamental yet non-trivial question about LMCs is whether two LMCs generate the same distribution on words.
This problem itself has applications in verification~\cite{KMOWW:CAV11} and can be solved in polynomial time using algorithms that are based on linear algebra~\cite{Schutzenberger,Paz71,CortesMRdistance}.
If two such distributions are not equal, one may ask how different they are.
There exist various \emph{distances} between discrete distributions, see, e.g., \cite[Section~3]{CortesMRdistance}.
One of them is the total variation distance (in the following just called \emph{distance}), which can be defined by $d(\pi_1, \pi_2) = \max_{W \subseteq \Sigma^*}|\pi_1(W) - \pi_2(W)|$ in the case of LMCs.
That is, $d(\pi_1, \pi_2)$ is the largest possible difference between probabilities that $\pi_1$ and~$\pi_2$ assign to the same set of words.
This distance is, up to a factor~$2$, equal to the $L_1$-norm of the difference between $\pi_1$ and~$\pi_2$, i.e., $2 d(\pi_1, \pi_2) = \sum_{w \in \Sigma^*} |\pi_1(w) - \pi_2(w)|$.
Clearly, $\pi_1$ and~$\pi_2$ are equal if and only if their distance is~$0$.

It is immediate from the definition of the distance that if $L$~is a family of LMCs whose pairwise distances are bounded by $b \ge 0$ then for any event~$W \subseteq \Sigma^*$ and any two LMCs $\M_1, \M_2 \in L$ we have $|\pi_1(W) - \pi_2(W)| \le b$.
From a verification point of view, this means that one needs to model check only one LMC in the family to obtain an approximation within~$b$ for the probabilities that the LMCs satisfy a given property~$W$.
Therefore, computing or approximating the distance can make model checking more efficient.
It is shown in~\cite{ChenBW12} that the \emph{bisimilarity pseudometric} defined in~\cite{DesharnaisGJPmetrics} is an upper bound on the total variation distance and can be computed in polynomial time.
The bisimilarity pseudometric has more direct bearings on \emph{branching-time} system properties, which, in addition to emitted labels, take LMC states into account (not considered in this paper).

The problem of computing the distance was first studied in~\cite{LyngsoP02}: they show that computing the distance is NP-hard.
In~\cite{CortesMRdistance} it was shown that even approximating the distance within an $\varepsilon>0$ given in binary is NP-hard.
In this paper we improve these results.
We show that it is undecidable whether the distance is greater than a given threshold.
Further we show that approximating the distance is \#P-hard and in PSPACE.
The \#P-hardness construction is relatively simple, perhaps simpler than the construction underlying the NP-hardness result in~\cite{CortesMRdistance}.
In contrast, our PSPACE algorithm requires a combination of special techniques: rounding-error analysis in floating-point arithmetic and Ladner's result~\cite{Ladner89} on counting in polynomial space.

\section{Preliminaries} \label{sec-prelim}


Let $Q$ be a finite set.
We view elements of $\R^Q$ as \emph{vectors}, more specifically as row vectors.
We write $\vec{1}$ for the all-1 vector, i.e., the element of $\{1\}^Q$.
For a vector $\mu \in \R^Q$, we denote by~$\mu^\top$ its transpose, a column vector.
A vector $\mu \in [0,1]^Q$ is a \emph{distribution} 
\emph{over~$Q$} if $\mu \vec{1}^\top = 1$. 
For $q \in Q$ we write~$\delta_q$ for the (\emph{Dirac}) distribution over~$Q$ with $\delta_q(q) = 1$ and $\delta_q(r) = 0$ for $r \in Q \setminus\{q\}$.
We view elements of $\R^{Q \times Q}$ as \emph{matrices}.
A matrix $M \in [0,1]^{Q \times Q}$ is called \emph{stochastic} if each row sums up to one, i.e.,
 $M \vec{1}^\top = \vec{1}^\top$.

\begin{definition}
A \emph{labelled (discrete-time, finite-state, finite-word) Markov chain (LMC)} is a quadruple $\M = (Q, \Sigma, M, \eta)$ where
 $Q$ is a finite set of states,
 $\Sigma$ is a finite alphabet of labels,
 the mapping $M: \Sigma \to [0,1]^{Q \times Q}$ specifies the transitions,
 and $\eta \in [0,1]^Q$, with $\eta^\top + \sum_{a \in \Sigma} M(a) \vec{1}^\top = \vec{1}^\top$, specifies the end-of-word probability of each state.
\end{definition}
Intuitively, if the LMC is in state~$q$, then with probability~$M(a)(q,q')$ it emits~$a$ and moves to state~$q'$, and with probability~$\eta(q)$ it stops emitting labels.
For the complexity results in this paper, we assume that all numbers in~$\eta$ and in the matrices~$M(a)$ for $a \in \Sigma$ are rationals given as fractions of integers represented in binary.
We extend $M$ to the mapping $M: \Sigma^* \to [0,1]^{Q \times Q}$ with $M(a_1 \cdots a_k) = M(a_1) \cdots M(a_k)$ for $a_1, \ldots, a_k \in \Sigma$.
Intuitively, if the LMC is in state~$q$ then with probability~$M(w)(q,q')$ it emits the word~$w \in \Sigma^*$ and moves (in $|w|$ steps) to state~$q'$.
We require that each state of an LMC have a positive-probability path to some state~$q$ with $\eta(q) > 0$.

Fix an LMC~$\M = (Q, \Sigma, M, \eta)$ for the rest of this section.
%
To an (initial) distribution $\pi$ over~$Q$ we associate the discrete probability space $(\Sigma^*,2^{\Sigma^*},\Pr_\pi)$ with $\Pr_\pi(w) := \Pr_\pi(\{w\}) := \pi M(w) \eta^\top$.
To avoid clutter and when confusion is unlikely, we may identify the 
 distribution~$\pi \in [0,1]^Q$ with its induced 
 probability measure~$\Pr_\pi$;
 i.e., for a word or set of words~$W$ we may write $\pi(W)$ instead of~$\Pr_\pi(W)$.

Given two initial distributions $\pi_1, \pi_2$,
 the \emph{(total variation) distance} between $\pi_1$ and~$\pi_2$ is defined as follows:%
\footnote{One could analogously define the total variation distance between two LMCs $\M_1 = (Q_1, \Sigma, M_1, \eta_1)$ and $\M_2 = (Q_2, \Sigma, M_2, \eta_2)$
with initial distributions $\pi_1$ and~$\pi_2$ over $Q_1$ and~$Q_2$, respectively.
Our definition is without loss of generality, as one can take the LMC $\M = (Q, \Sigma, M, \eta)$ where $Q$ is the disjoint union of $Q_1$ and~$Q_2$,
and $M, \eta$ are defined using $M_1, M_2, \eta_1, \eta_2$ in the straightforward manner.}
\[
 d(\pi_1, \pi_2) \ := \ \sup_{W \subseteq \Sigma^*} |\pi_1(W) - \pi_2(W)| \,.
\]
As $\pi_1(W) - \pi_2(W) = \pi_2(\Sigma^* \setminus W) - \pi_1(\Sigma^* \setminus W)$,
 we have $d(\pi_1, \pi_2) = \sup_{W \subseteq \Sigma^*} (\pi_1(W) - \pi_2(W))$.
The following proposition follows from basic principles, see, e.g.,~\cite[Lemma~11.1]{ProbAndComp}.
In particular, it says that the supremum is attained and the total variation distance is closely related to the $L_1$-distance:



\begin{proposition} \label{prop-tv-dist-basic}
Let $\M$ be an LMC.
For any two initial distributions $\pi_1, \pi_2$ we have:
\[
d(\pi_1, \pi_2)
 \ = \ \max_{W \subseteq \Sigma^*} (\pi_1(W) - \pi_2(W))
 \ = \ \frac12 \sum_{w \in \Sigma^*} | \pi_1(w) - \pi_2(w) |
\]
The maximum is attained by $W = \{w \in \Sigma^* : \pi_1(w) \ge \pi_2(w) \}$.
\end{proposition}
In view of this proposition, all complexity results on the (total variation) distance hold equally for the $L_1$-distance.

An LMC~$\M$ is called \emph{acyclic} if its transition graph is acyclic.
Equivalently, $\M$~is acyclic if for all $q \in Q$ we have that $\Pr_{\delta_q}$ has finite support, i.e., $\{w \in \Sigma^* : \Pr_{\delta_q}(w)>0\}$ is finite.

\section{The Threshold-Distance Problem} \label{sec-threshold-distance}

In \cite[Section 6]{LyngsoP02} (see also \cite[Theorem 7]{CortesMRdistance}), a reduction is given from the \emph{clique} decision problem to show that computing the distance in LMCs is NP-hard.
In that reduction the distance is rational and its bit size polynomial in the input.
It was shown in~\cite[Proposition~12]{14CK-LICS} that the distance~$d$ can be irrational.
Define the \emph{non-strict (resp.\ strict) threshold-distance} problem as follows:
Given an LMC, two initial distributions $\pi_1, \pi_2$, and a threshold $\tau \in [0,1] \cap \Q$,
 decide whether $d(\pi_1, \pi_2) \ge \tau$ (resp.\ $d(\pi_1, \pi_2) > \tau$).
In~\cite[Proposition 14]{14CK-LICS} it was shown that the non-strict threshold-distance problem is NP-hard with respect to Turing reductions.

In the following two subsections we consider the threshold-distance problem for general and acyclic LMCs, respectively.

\subsection{General LMCs} \label{sub-threshold-distance-words}
We show:
\begin{theorem} \label{thm-threshold-distance-words}
The strict threshold-distance problem is undecidable.
\end{theorem}
\begin{proof}
We reduce from the emptiness problem for probabilistic automata.
A \emph{probabilistic automaton} is a tuple $\A = (Q, \Sigma, M, \alpha, F)$ where
 $Q$ is a finite set of states,
 $\Sigma$ is a finite alphabet of labels,
 the mapping $M: \Sigma \to [0,1]^{Q \times Q}$, where $M(a)$ is a stochastic matrix for each $a \in \Sigma$, specifies the transitions,
 $\alpha \in [0,1]^Q$ is an initial distribution,
 and $F \subseteq Q$ is a set of accepting states.
Extend~$M$ to $M: \Sigma^* \to [0,1]^{Q \times Q}$ as in the case of LMCs.
In the case of a probabilistic automaton, $M(w)$ is a stochastic matrix for each $w \in \Sigma^*$.
For each $w \in \Sigma^*$ define $\Pr_{\A}(w) := \alpha M(w) \eta^\top$ where $\eta \in \{0,1\}^Q$ denotes the characteristic vector of~$F$.
The probability~$\Pr_{\A}(w)$ can be interpreted as the probability that $\A$ accepts~$w$, i.e., the probability that after inputting~$w$ the automaton~$\A$ is in an accepting state.
The \emph{emptiness problem} asks, given a probabilistic automaton~$\A$, whether there is a word $w \in \Sigma^*$ such that $\Pr_{\A}(w) > \frac12$.
This problem is known to be undecidable~\cite[p.~190, Theorem~6.17]{Paz71}.

In the following we assume $\Sigma = \{a_1, \ldots, a_k\}$.
Given a probabilistic automaton~$\A$ as above, construct an LMC $\M = (Q \cup \{q_1, q_\eow\}, \Sigma \cup \{b, \acc, \rej\}, \overline{M}, \delta_{q_\eow})$ such that $q_1, q_\eow$ are fresh states, and $b, \acc, \rej$ are fresh labels.
The transitions originating in the fresh states~$q_1, q_\eow$ are as follows:
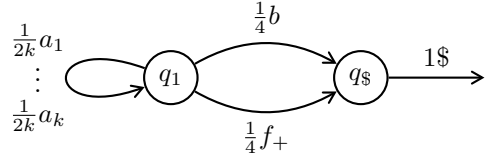
\begin{center}
\begin{tikzpicture}[scale=2.5,LMC style]
\node[state] (q1) at (0,0) {$q_1$};
\node[state] (qeow) at (+1,0) {$q_\eow$};
\node (end) at (1.7,0) {};
\path[->] (q1) edge [loop,out=160,in=200,looseness=15] (q1);
\node at (-0.7,0.2) {$\frac{1}{2 k} a_1$};
\node at (-0.7,0.03) {$\vdots$};
\node at (-0.7,-0.2) {$\frac{1}{2 k} a_k$};
\path[->] (q1) edge [bend left] node[above] {$\frac14 b$} (qeow);
\path[->] (q1) edge [bend right] node[below] {$\frac14 \acc$} (qeow);
\path[->] (qeow) edge node[above] {$1 \eow$} (end);
\end{tikzpicture}
\end{center}
Here and in the subsequent pictures we use a convention that there be a state $q_\eow$ with $\eta(q_\eow) = 1$ and that $\eta(q) = 0$ hold for all other states.

Define $\pi_1 := \delta_{q_1}$.
Then for all $w \in \Sigma^*$ we have:
\begin{equation}
 \pi_1(w b) \ = \ \pi_1(w \acc) \ = \ \left(\frac1{2 k}\right)^{|w|} \cdot \frac14 \label{eq-strict-threshold-distance-1}
\end{equation}
The transitions originating in the states in~$Q$ are defined so that all $q \in Q$ emit each $a \in \Sigma$ with probability~$\frac1{2 k}$ (like~$q_1$).
For all $q \in F$ there is a transition to~$q_\eow$ labelled with $\frac12$ and~$\acc$;
for all $q \in Q \setminus F$ there is a transition to~$q_\eow$ labelled with $\frac12$ and~$\rej$:
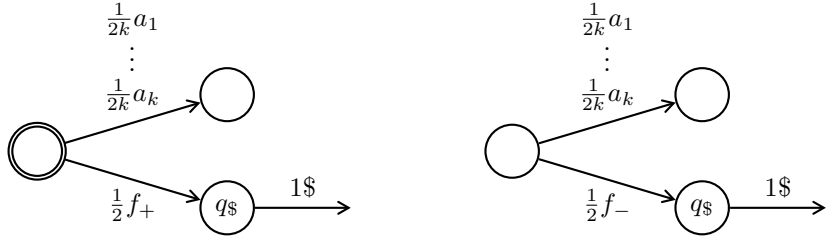
\begin{center}
\begin{tikzpicture}[scale=2.5,LMC style]
\node[state,accepting] (qa) at (0,0) {};
\node[state] (qa1) at (1,0.3) {};
\node[state] (qa3) at (1,-0.3) {$q_\eow$};
\node (enda) at (1.7,-0.3) {};
\path[->] (qa) edge (qa1);
\node at (0.5,0.70) {$\frac{1}{2 k} a_1$};
\node at (0.5,0.53) {$\vdots$};
\node at (0.5,0.30) {$\frac{1}{2 k} a_k$};
\path[->] (qa) edge node[below] {$\frac12 \acc$} (qa3);
\path[->] (qa3) edge node[above] {$1 \eow$} (enda);

\node[state] (qr) at (2.5,0) {};
\node[state] (qr1) at (3.5,0.3) {};
\node[state] (qr3) at (3.5,-0.3) {$q_\eow$};
\node (endr) at (4.2,-0.3) {};
\path[->] (qr) edge (qr1);
\node at (3,0.70) {$\frac{1}{2 k} a_1$};
\node at (3,0.53) {$\vdots$};
\node at (3,0.30) {$\frac{1}{2 k} a_k$};
\path[->] (qr) edge node[below] {$\frac12 \rej$} (qr3);
\path[->] (qr3) edge node[above] {$1 \eow$} (endr);
\end{tikzpicture}
\end{center}
Formally, for $q, r \in Q$ and $a \in \Sigma$ set $\overline{M}(a)(q,r) := \frac1{2 k} M(a)(q,r)$.
For $q \in F$ set $\overline{M}(\acc)(q,q_\eow) := \frac12$, and
for $q \in Q \setminus F$ set $\overline{M}(\rej)(q,q_\eow) := \frac12$.
Define $\pi_2 := \alpha$ (in the natural way, i.e., with $\pi_2(q_1) = \pi_2(q_\eow) = 0$).
Then for all $w \in \Sigma^*$ we have:
\begin{equation} \label{eq-strict-threshold-distance-2}
\begin{aligned}
 \pi_2(w \acc) &\ = \ \left(\frac1{2 k}\right)^{|w|} \cdot \Pr\nolimits_{\A}(w) \cdot \frac12 \qquad \text{and} \\
 \pi_2(w \rej) &\ = \ \left(\frac1{2 k}\right)^{|w|} \cdot (1 - \Pr\nolimits_{\A}(w)) \cdot \frac12
\end{aligned}
\end{equation}
Consider $L := \Sigma^* \{b, \acc\}$.
We have $\pi_1(L) = 1$.
One can compute $\pi_2(L)$ in polynomial time by computing the probability of reaching a transition labelled by~$\acc$ (the label~$b$ is not reachable).
We claim that there is $w \in \Sigma^*$ with $\Pr_{\A}(w) > \frac12$ if and only if $d(\pi_1, \pi_2) > \pi_1(L) - \pi_2(L)$.
It remains to prove this claim.

Suppose there is no $w \in \Sigma^*$ with $\Pr_{\A}(w) > \frac12$.
Then, by \eqref{eq-strict-threshold-distance-1} and~\eqref{eq-strict-threshold-distance-2}, for all $w \in \Sigma^*$ we have $\pi_1(w \acc) \ge \pi_2(w \acc)$.
Hence:
\[
\left\{w \in (\Sigma \cup \{b, \acc, \rej\})^* : \pi_1(w) > 0,\ \pi_1(w) \ge \pi_2(w) \right\}
\quad = \quad L
\]
By Proposition~\ref{prop-tv-dist-basic} it follows $d(\pi_1, \pi_2) = \pi_1(L) - \pi_2(L)$.

Conversely, suppose there is $w \in \Sigma^*$ with $\Pr_{\A}(w) > \frac12$.
Consider $L' := L \setminus \{w \acc\}$.
We have:
\begin{align*}
d(\pi_1, \pi_2)
& \ \ge \ \pi_1(L') - \pi_2(L') && \text{Proposition~\ref{prop-tv-dist-basic}} \\
& \  =  \ \pi_1(L) - \pi_1(w \acc) - \pi_2(L) + \pi_2(w \acc) && \text{definition of~$L'$} \\
& \  =  \ \pi_1(L) - \pi_2(L) + \left(\frac1{2 k}\right)^{|w|} \cdot \left(\frac12 \Pr\nolimits_{\A}(w) - \frac14 \right) && \text{by \eqref{eq-strict-threshold-distance-1} and~\eqref{eq-strict-threshold-distance-2}} \\
& \  >  \ \pi_1(L) - \pi_2(L) && \Pr\nolimits_{\A}(w) > \frac12 \hfill \qedhere
\end{align*}
\end{proof}
Cortes, Mohri, and Rastogi~\cite{CortesMRdistance} conjectured ``that the problem of computing the [\ldots] distance [\ldots] is in fact undecidable'', see the discussion after the proof of~\cite[Theorem~7]{CortesMRdistance}.
Theorem~\ref{thm-threshold-distance-words} proves \emph{one interpretation} of that conjecture.
But the distance can be approximated with arbitrary precision, cf.\ Section~\ref{sec-approximation}, so the distance is ``computable'' in this sense.

In~\cite[Theorem 15]{14CK-LICS} it was shown that there is a polynomial-time many-one reduction from the square-root-sum problem to the non-strict threshold-distance problem for LMCs.
Decidability of the non-strict threshold-distance problem remains open.

\subsection{Acyclic LMCs} \label{sub-threshold-distance-acyclic}

It was shown in~\cite[Section 6]{LyngsoP02} and \cite[Proposition 14]{14CK-LICS} that the non-strict threshold-distance problem is NP-hard with respect to Turing reductions, even for acyclic LMCs.
We improve this result to PP-hardness:
\begin{restatable}{propositionS}{propthresholddistanceacyclic} \label{prop-threshold-distance-acyclic}
The non-strict and strict threshold-distance problems are PP-hard, even for acyclic LMCs and even with respect to many-one reductions.
\end{restatable}
\noindent The proof uses the connection between PP and \#P.
Consider the problem \#NFA, which is defined as follows: given a nondeterministic finite automaton (NFA)~$\A$ over alphabet~$\Sigma$, and a number $n \in \N$ in unary, compute $|L(\A) \cap \Sigma^n|$, i.e., the number of accepted words of length~$n$.
The problem \#NFA is \#P-complete~\cite{KSM-counting-strings}.
The following lemma forms the core of the proof of Proposition~\ref{prop-threshold-distance-acyclic}:
\begin{lemma} \label{lem-non-strict-threshold-distance}
Given an NFA $\A = (Q, \Sigma, \delta, q^{(1)}, F)$ and a number $n \in \N$ in unary, one can compute in polynomial time an acyclic LMC~$\M$ and initial distributions $\pi_1, \pi_2$ and a rational number~$y$ such that
\[
 d(\pi_1, \pi_2) \ = \ y + \frac{|\Sigma^n \setminus L(\A)|}{|\Sigma|^n |Q|^n} \,.
\]
\end{lemma}
\begin{proof}
In the following we assume $Q = \{q^{(1)}, \ldots, q^{(s)}\}$ and $\Sigma = \{a_1, \ldots, a_k\}$.
Construct the acyclic LMC $\M = (Q', \Sigma \cup \{b, \acc, \rej\}, M)$ such that
\[
 Q' \ = \ \{p_0, p_1, \ldots, p_n, q_\eow\} \ \cup \ \{ q_i^{(j)} : 0 \le i \le n,\ 1 \le j \le s \}
     \ \cup \ \{ r_i : 0 \le i \le n \}
\]
and $b, \acc, \rej$ are fresh labels.
The transitions and end-of-word probabilities originating in the states $p_0, \ldots, p_{n}, q_\eow$ are as follows:
\begin{center}
\begin{tikzpicture}[scale=2.5,LMC style]
\node[state] (p0) at (0,0) {$p_0$};
\node[state] (p1) at (1,0) {$p_1$};
\node[state,draw=none] (pd) at (2,0) {$\ldots$};
\node[state] (pn) at (3,0) {$p_n$};
\node[state] (qeow) at (+4,0) {$q_\eow$};
\node (end) at (4.7,0) {};
\path[->] (p0) edge (p1);
\path[->] (p1) edge (pd);
\path[->] (pd) edge (pn);
\path[->] (pn) edge [bend left] node[above] {$\frac1{s^n} b$} (qeow);
\path[->] (pn) edge [bend right] node[below] {$(1 - \frac1{s^n}) \rej$} (qeow);
\path[->] (qeow) edge node[above] {$1 \eow$} (end);
\node at (0.5,0.5) {$\frac{1}{k} a_1$};
\node at (0.5,0.34) {$\vdots$};
\node at (0.5,0.12) {$\frac{1}{k} a_k$};
\node at (1.5,0.5) {$\frac{1}{k} a_1$};
\node at (1.5,0.34) {$\vdots$};
\node at (1.5,0.12) {$\frac{1}{k} a_k$};
\node at (2.5,0.5) {$\frac{1}{k} a_1$};
\node at (2.5,0.34) {$\vdots$};
\node at (2.5,0.12) {$\frac{1}{k} a_k$};
\end{tikzpicture}
\end{center}
Define $\pi_1 := \delta_{p_0}$.
Then for all $w \in \Sigma^n$ we have:
\begin{align}
 \pi_1(w b) \ &= \ \frac1{k^n} \cdot \frac1{s^n} \label{eq-non-strict-threshold-distance-p1-b} \\
 \pi_1(w \rej) \ &= \ \frac1{k^n} \cdot \Big(1 - \frac{1}{s^n}\Big) \label{eq-non-strict-threshold-distance-p1-rej}
\end{align}
The transitions originating in the states $q_i^{(j)}, r_i$ are as follows.
For each $a \in \Sigma$ and each $i \in \{0, \ldots, n-1\}$ set:
\begin{align*}
M(a)\big(q_i^{(j)}, q_{i+1}^{(j')}\big) & \ := \ \frac{1}{k} \cdot \frac{1}{s} && \forall\, j \in \{1, \ldots, s\} \quad \forall\,q^{(j')} \in \delta(q^{(j)},a)\\
M(a)\big(q_i^{(j)}, r_{i+1}\big)        & \ := \ \frac{1}{k} \cdot \Big( 1 - \frac{|\delta(q^{(j)},a)|}{s} \Big) && \forall\, j \in \{1, \ldots, s\} \\
M(a)\big(r_i, r_{i+1}\big)              & \ := \ \frac{1}{k}
\end{align*}
Observe that if $i \in \{0, \ldots, n-1\}$ then $r_i$ and all $q_i^{(j)}$ emit each $a \in \Sigma$ with probability~$1/k$.
For each $q^{(j)} \in F$ set $M(\acc)(q_n^{(j)}, q_\eow) := 1$.
For each $q^{(j)} \not\in F$ set $M(\rej)(q_n^{(j)}, q_\eow) := 1$.
Finally, set $M(\rej)(r_n, q_\eow) := 1$.
\begin{example}
\it
We illustrate this construction with the following NFA~$\A$ over $\Sigma = \{a_1, a_2\}$:
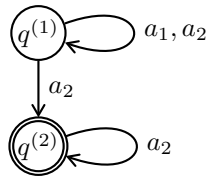
\begin{center}
\begin{tikzpicture}[scale=2.5,LMC style]
\node[state] (q1) at (0,0) {$q^{(1)}$};
\node[state,accepting] (q2) at (0,-0.6) {$q^{(2)}$};
\path[->] (q1) edge [loop,out=20,in=-20,looseness=13] node[right] {$a_1,a_2$} (q1);
\path[->] (q1) edge node[right] {$a_2$} (q2);
\path[->] (q2) edge [loop,out=20,in=-20,looseness=13] node[right] {$a_2$} (q2);
\end{tikzpicture}
\end{center}
For $n=3$ we obtain the following transitions: 
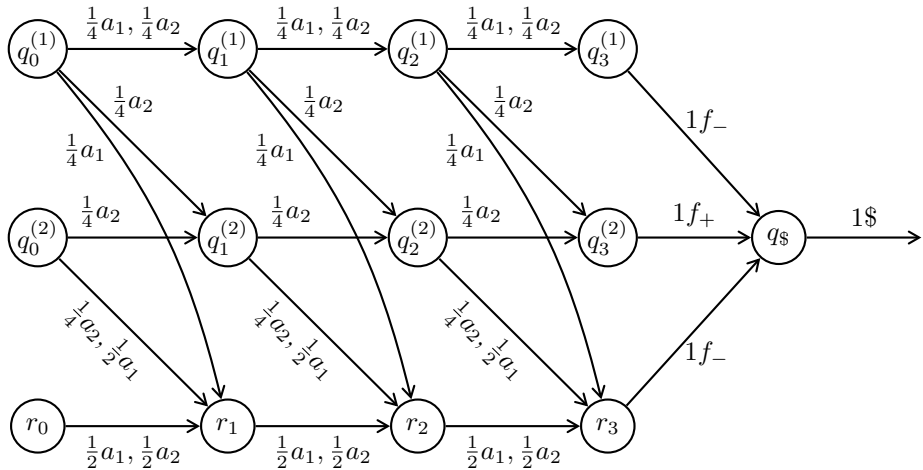
\begin{center}
\begin{tikzpicture}[scale=2.5,LMC style]
\node[state] (q01) at (0,0)  {$q_0^{(1)}$};
\node[state] (q02) at (0,-1) {$q_0^{(2)}$};
\node[state] (r0)  at (0,-2) {$r_0$};
\node[state] (q11) at (1,0)  {$q_1^{(1)}$};
\node[state] (q12) at (1,-1) {$q_1^{(2)}$};
\node[state] (r1)  at (1,-2) {$r_1$};
\node[state] (q21) at (2,0)  {$q_2^{(1)}$};
\node[state] (q22) at (2,-1) {$q_2^{(2)}$};
\node[state] (r2)  at (2,-2) {$r_2$};
\node[state] (q31) at (3,0)  {$q_3^{(1)}$};
\node[state] (q32) at (3,-1) {$q_3^{(2)}$};
\node[state] (r3)  at (3,-2) {$r_3$};
\node[state] (qeow) at (3.9,-1) {$q_\eow$};
\node (end) at (4.7,-1) {};
\path[->] (qeow) edge node[above] {$1 \eow$} (end);
\path[->] (q01) edge node[above] {$\frac14 a_1,\frac14 a_2$} (q11);
\path[->] (q01) edge node[above,yshift=6] {$\frac14 a_2$} (q12);
\path[->] (q01) edge[bend left=13] node[left,pos=0.25] {$\frac14 a_1$} (r1);
\path[->] (q02) edge node[above,pos=0.25] {$\frac14 a_2$} (q12);
\path[->] (q02) edge node[below,xshift=-5,sloped] {$\frac14 a_2, \frac12 a_1$} (r1);
\path[->] (r0)  edge node[below] {$\frac12 a_1, \frac12 a_2$} (r1);
\path[->] (q11) edge node[above] {$\frac14 a_1,\frac14 a_2$} (q21);
\path[->] (q11) edge node[above,yshift=6] {$\frac14 a_2$} (q22);
\path[->] (q11) edge[bend left=13] node[left,pos=0.25] {$\frac14 a_1$} (r2);
\path[->] (q12) edge node[above,pos=0.25] {$\frac14 a_2$} (q22);
\path[->] (q12) edge node[below,xshift=-5,sloped] {$\frac14 a_2, \frac12 a_1$} (r2);
\path[->] (r1)  edge node[below] {$\frac12 a_1, \frac12 a_2$} (r2);
\path[->] (q21) edge node[above] {$\frac14 a_1,\frac14 a_2$} (q31);
\path[->] (q21) edge node[above,yshift=6] {$\frac14 a_2$} (q32);
\path[->] (q21) edge[bend left=13] node[left,pos=0.25] {$\frac14 a_1$} (r3);
\path[->] (q22) edge node[above,pos=0.25] {$\frac14 a_2$} (q32);
\path[->] (q22) edge node[below,xshift=-5,sloped] {$\frac14 a_2, \frac12 a_1$} (r3);
\path[->] (r2)  edge node[below] {$\frac12 a_1, \frac12 a_2$} (r3);
\path[->] (q31) edge node[above,xshift=5] {$1 \rej$} (qeow);
\path[->] (r3)  edge node[below,xshift=5] {$1 \rej$} (qeow);
\path[->] (q32) edge node[above] {$1 \acc$} (qeow);
\end{tikzpicture}
\end{center}
\end{example}
Define $\pi_2 := \delta_{q_0^{(1)}}$.
For all $w \in \Sigma^*$ write $\nacc(w)$ for the number of accepting $w$-labelled runs of the automaton~$\A$, i.e., the number of $w$-labelled paths from $q^{(1)}$ to a state in~$F$.
For all $w \in \Sigma^n$ we have:
\begin{align}
 \pi_2(w \acc) \ &= \ \frac{1}{k^n} \cdot \frac{\nacc(w)}{s^n} \label{eq-non-strict-threshold-distance-p2-b} \\
 \pi_2(w \rej) \ &= \ \frac1{k^n} \cdot \Big(1 - \frac{\nacc(w)}{s^n}\Big) \label{eq-non-strict-threshold-distance-p2-rej}
\end{align}
Define $B := \Sigma^n \{b, \rej\}$.
By \eqref{eq-non-strict-threshold-distance-p1-b},~\eqref{eq-non-strict-threshold-distance-p1-rej} we have $\pi_1(B) = 1$.
One can compute~$\pi_2(B)$ in polynomial time by computing the probability of reaching a transition labelled by~$\rej$ (the label~$b$ is not reachable).
Set $y := \pi_1(B) - \pi_2(B)$.

It follows from Proposition~\ref{prop-tv-dist-basic} that $d(\pi_1, \pi_2) = \pi_1(L) - \pi_2(L)$ holds for
\newcommand{\Lcomp}{\overline{L(\A)}}%
\begin{align*}
L &\ := \ \{ w \in (\Sigma \cup \{b,\acc,\rej\})^* : 0 < \pi_1(w) \ge \pi_2(w) \}\,.
\intertext{Observe that $L(\A) = \{w \in \Sigma^n : \nacc(w) \ge 1\}$.
Hence it follows with \eqref{eq-non-strict-threshold-distance-p1-b},~\eqref{eq-non-strict-threshold-distance-p1-rej},~\eqref{eq-non-strict-threshold-distance-p2-rej}:}
L &\ = \  \Sigma^n \{b\} \; \cup \; (\Sigma^n \cap L(\A)) \{\rej\}\
\intertext{Defining $\Lcomp := \Sigma^n \setminus L(\A)$ we can write:}
L &\ = \  B \setminus \big( \Lcomp \{\rej\} \big)
\end{align*}
Thus we have:
\begin{align*}
d(\pi_1, \pi_2)
\ &= \ \pi_1\Big( B \setminus \big( \Lcomp \{\rej\} \big) \Big) - \pi_2\Big( B \setminus \big( \Lcomp \{\rej\} \big) \Big) && \text{as argued above} \\
\ &= \ y + \pi_2\big( \Lcomp \{\rej\} \big) - \pi_1\big( \Lcomp \{\rej\} \big) && \text{definition of~$y$}
\intertext{Observe that $\Lcomp = \{w \in \Sigma^n : \nacc(w) = 0\}$. Hence we can continue:}
\ &= \ y + \frac{\big|\Lcomp\big|}{k^n} - \frac{\big|\Lcomp\big|}{k^n} \cdot \Big(1 - \frac1{s^n}\Big) && \text{by \eqref{eq-non-strict-threshold-distance-p2-rej},~\eqref{eq-non-strict-threshold-distance-p1-rej}} \\
\ &= \ y + \frac{\big|\Lcomp\big|}{k^n s^n} \ = \ y + \frac{|\Sigma^n \setminus L(\A)|}{|\Sigma|^n |Q|^n} && \text{definitions} \hfill \qedhere
\end{align*}
\end{proof}
The PP lower bound from Proposition~\ref{prop-threshold-distance-acyclic} is tight for acyclic LMCs:
\begin{restatable}{theorem}{thmthresholddistanceacyclic} \label{thm-threshold-distance-acyclic}
The non-strict and strict threshold-distance problems are PP-complete for acyclic LMCs.
\end{restatable}
\begin{remarkS} \label{rem-dk}
The works \cite{LyngsoP02,CortesMRdistance} also consider the $L_k$-distances for integers~$k$:
\[
 d_k(\pi_1, \pi_2) \ := \ \sum_{w \in \Sigma^*} | \pi_1(w) - \pi_2(w) |^k
\]
For any fixed even~$k$ one can compute~$d_k$ in polynomial time, see, e.g., \cite[Theorem 6]{CortesMRdistance}.
In contrast, it is NP-hard to compute or even approximate~$d_k$ for any odd~$k$ \cite[Theorems 7 and 10]{CortesMRdistance}.
Our PP- and \#P-hardness results (Proposition~\ref{prop-threshold-distance-acyclic} and Theorem~\ref{thm-approx-distance}) hold for~$d_1$ (due to Proposition~\ref{prop-tv-dist-basic}) but the reductions do not apply in an obvious way to~$d_k$ for any $k \ge 2$.
However, the argument in the proof of Theorem~\ref{thm-threshold-distance-acyclic} for the PP upper bound does generalize to all~$d_k$, see \iftechrep{Appendix~\ref{app-thm-threshold-distance-acyclic}}{\cite{18K:ICALP-TR}}.
\end{remarkS}

\section{Approximation} \label{sec-approximation}

As the strict threshold-distance problem is undecidable (Theorem~\ref{thm-threshold-distance-words}), 
one may ask whether the distance can be approximated.
It is not hard to see that the answer is yes.
In fact, it was shown in~\cite[Corollary~8]{14CK-LICS} that the distance can be approximated within an arbitrary additive error even for \emph{infinite-word} LMCs, but no complexity bounds were given.
In this section we provide bounds on the complexity of approximating the distance for (finite-word) LMCs.

\subsection{Hardness} \label{sub-approximation-hardness}

Lemma~\ref{lem-non-strict-threshold-distance} implies hardness of approximating the distance:
\begin{theorem} \label{thm-approx-distance}
Given an LMC and initial distributions $\pi_1, \pi_2$ and an error bound $\varepsilon > 0$ in binary, it is \#P-hard to  compute a number~$x$ with $|d(\pi_1, \pi_2) - x| \le \varepsilon$, even for acyclic LMCs.
\end{theorem}
\begin{proof}
Recall that the problem \#NFA is \#P-complete~\cite{KSM-counting-strings}.
Let $\A$ be the given NFA and $n \in \N$.
Let $\M, \pi_1, \pi_2, y$ be as in Lemma~\ref{lem-non-strict-threshold-distance}.
Approximate $d(\pi_1, \pi_2)$ within $1/(3 |\Sigma|^n |Q|^n)$ and call the approximation~$\tilde d$.
It follows from Lemma~\ref{lem-non-strict-threshold-distance} that $|L(\A) \cap \Sigma^n|$ is the unique integer~$u$ with
\[
 \left| y + \frac{|\Sigma|^n - u}{|\Sigma|^n |Q|^n} - \tilde{d} \right| \ \le \  \frac{1}{3 |\Sigma|^n |Q|^n}\,.
\]
Such~$u$ can be computed in polynomial time.
\end{proof}
Theorem~\ref{thm-approx-distance} improves the NP-hardness result of~\cite[Proposition 9]{14CK-LICS}.
In fact, PP and~\#P are substantially harder than~NP:
By Toda's theorem~\cite{To91}, the polynomial-time hierarchy (PH) is contained in $\mathrm{P^{PP} = P^{\#P}}$.
Therefore, any problem in PH can be decided in deterministic polynomial time with the help of an oracle for the threshold-distance problem or for approximating the distance.

\subsection{Acyclic LMCs} \label{sub-approximation-acyclic}

Towards approximation algorithms, define $W_2 := \{w \in \Sigma^* : \pi_1(w) \ge \pi_2(w) \}$ and $W_1 := \{w \in \Sigma^* : \pi_1(w) < \pi_2(w) \}$.
By Proposition~\ref{prop-tv-dist-basic} we have:
\begin{align}
d(\pi_1, \pi_2) \ & = \ \pi_1(W_2) - \pi_2(W_2) \ = \ 1 - \pi_1(W_1) - \pi_2(W_2) \label{eq-W1-W2}
\end{align}
Therefore, to approximate $d(\pi_1,\pi_2)$ it suffices to approximate $\pi_i(W_i)$.
A simple sampling scheme leads to the following theorem:
\begin{theorem} \label{thm-approx-sampling}
There is a randomized algorithm, $R$, that, given an acyclic LMC~$\M$ and initial distributions $\pi_1, \pi_2$ and an error bound~$\varepsilon>0$ and an error probability $\delta \in (0,1)$, does the following:
\begin{itemize}
\item $R$ computes, with probability at least~$1 - \delta$, a number~$x$ with $|d(\pi_1,\pi_2) - x| \le \varepsilon$;
\item $R$ runs in time polynomial in $\frac{\log \delta}{\varepsilon}$ and in the encoding size of $\M$ and~$\pi_1,\pi_2$.
\end{itemize}
\end{theorem}
Note that $\frac{1}{\varepsilon}$ is not polynomial in the bit size of~$\varepsilon$, so combining Theorems \ref{thm-approx-distance} and~\ref{thm-approx-sampling} does not imply breakthroughs in computational complexity.
\begin{proof}
Let $i \in \{1,2\}$.
The length of a longest word~$w$ with $\pi_i(w) > 0$ is polynomial in the encoding of the (acyclic) LMC~$\M$.
Thus, one can sample, in time polynomial in the encoding of $\M, \pi_1, \pi_2$, a word~$w$ according to~$\Pr_{\pi_i}$; i.e., any~$w$ is sampled with probability~$\pi_i(w)$.
Similarly, one can check in polynomial time whether $w \in W_i$.
If $m$~samples are taken, the proportion, say $\hat{p}_i$, of samples~$w$ such that $w \in W_i$ is an estimation of $\pi_i(W_i)$.
By Hoeffding's inequality, we have $|\hat{p}_i - \pi_i(W_i)| \ge \varepsilon/2$ with probability at most $2 e^{- m \varepsilon^2/2}$.
Choose $m \ge - \frac{2}{\varepsilon^2} \ln \frac{\delta}{4}$.
It follows that $|\hat{p}_i - \pi_i(W_i)| > \varepsilon/2$ with probability at most $\delta/2$.
Therefore, by~\eqref{eq-W1-W2}, the algorithm that returns $1 - \hat{p}_1 - \hat{p}_2$ has the required properties.
\end{proof}
 
\subsection{General LMCs} \label{sub-approximation-LMCs}

Finally we aim at an algorithm that approximates the distance within~$\varepsilon$, for $\varepsilon$ given in binary.
By Theorem~\ref{thm-approx-distance} such an algorithm cannot run in polynomial time unless P = PP.
For LMCs that are not necessarily acyclic, words of polynomial length may have only small probability,
so sampling approaches need to sample words of exponential length.
Thus, a naive extension of the algorithm from Theorem~\ref{thm-approx-sampling} leads to a randomized exponential-time algorithm.
We will develop a non-randomized PSPACE algorithm, resulting in the following theorem:
\begin{theorem} \label{thm-approximation-PSPACE}
Given an LMC, and initial distributions $\pi_1, \pi_2$, and an error bound $\varepsilon > 0$ in binary, one can compute in PSPACE a number~$x$ with $|d(\pi_1, \pi_2) - x| \le \varepsilon$.
\end{theorem}
The approximation algorithm combines special techniques.
The starting point is again the expression for the distance in~\eqref{eq-W1-W2}.
The following lemma allows the algorithm to neglect words that are longer than exponential:

\begin{restatable}{lemmaS}{lemapproximationlengthbound} \label{lem-approximation-length-bound}
Given an LMC, and initial distributions $\pi_1, \pi_2$, and a rational number $\lambda > 0$ in binary, one can compute in polynomial time a number $n \in \N$ in binary such that
\[
\pi_i(\Sigma^{> n})  \ \le \ \lambda \qquad \text{for both } i \in \{1,2\}\,.
\]
\end{restatable}
\newcommand{\tpi}{\widetilde{\pi}}%
\newcommand{\tW}{\widetilde{W}}%
\noindent For $n$ as in Lemma~\ref{lem-approximation-length-bound} and both $i \in \{1,2\}$, define $W_i' := W_i \cap \Sigma^{\le n}$.
By Lemma~\ref{lem-approximation-length-bound} it would suffice to approximate $\pi_i(W_i')$ for both~$i$, as we have by~\eqref{eq-W1-W2}:
\begin{equation}
\pi_1(W_1') + \pi_2(W_2') \quad \le \quad  1 - d(\pi_1, \pi_2) \quad \le \quad \pi_1(W_1') + \pi_2(W_2') + 2 \lambda
\label{eq-dist-Wiprime}
\end{equation}
However, it not obvious if $\pi_i(W_i')$ can be approximated efficiently, as for exponentially long words~$w$ it is hard to check if $w \in W_i'$ holds. 
Indeed, $\pi_i(w)$ may be very small and may have exponential bit size.
The main trick of our algorithm will be to approximate~$\pi_i(w)$ using floating-point arithmetic with small \emph{relative} error, say $\tpi_i(w) \in [ \pi_i(w) (1 - \theta), \pi_i(w) (1 + \theta)]$ for small $\theta>0$.
This allows us to approximate $\pi_1(W_1') + \pi_2(W_2')$ 
(crucially, not the two summands individually).
Indeed, define approximations for $W_1'$ and~$W_2'$ by
\[ 
 \tW_1 \ := \ \{w \in \Sigma^{\le n} : \tpi_1(w) < \tpi_2(w)\} \qquad \text{and} \qquad \tW_2 \ := \ \{w \in \Sigma^{\le n} : \tpi_1(w) \ge \tpi_2(w)\}\,.
\]
Then we have:
\begin{alignat*}{2}
 \pi_2(w) \ &\le \ \pi_1(w) \  <  \ \pi_2(w) + \theta \pi_1(w) + \theta \pi_2(w) && \qquad \text{for } w \in \tW_1 \cap W_2' \\
 \pi_1(w) \ & <  \ \pi_2(w) \ \le \ \pi_1(w) + \theta \pi_1(w) + \theta \pi_2(w) && \qquad \text{for } w \in \tW_2 \cap W_1'
\end{alignat*}
It follows:
\begin{equation}
\begin{aligned}
 \pi_2(\tW_1 \cap W_2') \ &\le \ \pi_1(\tW_1 \cap W_2') \ \le  \ \pi_2(\tW_1 \cap W_2') + 2 \theta  \\
 \pi_1(\tW_2 \cap W_1') \ &\le \ \pi_2(\tW_2 \cap W_1') \ \le  \ \pi_1(\tW_2 \cap W_1') + 2 \theta 
\end{aligned}
\label{eq-cross-terms}
\end{equation}
Hence we have:
\begin{align*}
\pi_1(W_1') + \pi_2(W_2') 
\ &=   \ \pi_1(\tW_1 \cap W_1') + \pi_2(\tW_1 \cap W_2') + \pi_2(\tW_2 \cap W_2') + \pi_1(\tW_2 \cap W_1') \\
\ &\mathop{\le}^\text{\eqref{eq-cross-terms}} \ \pi_1(\tW_1) + \pi_2(\tW_2) \\
\ &\mathop{\le}^\text{\eqref{eq-cross-terms}} \ \pi_1(W_1') + \pi_2(W_2') + 4 \theta
\end{align*}
By combining this with~\eqref{eq-dist-Wiprime} we obtain:
\begin{equation} \label{eq-approx-theta-epsilon}
\pi_1(\tW_1) + \pi_2(\tW_2) - 4 \theta \quad \le \quad  1 - d(\pi_1, \pi_2) \quad \le \quad \pi_1(\tW_1) + \pi_2(\tW_2) + 2 \lambda
\end{equation}
It remains to tie two loose ends:
\begin{enumerate}
\item develop a PSPACE method to approximate $\pi_i(w)$ within \emph{relative} error~$\theta$ for any $\theta > 0$ in binary, where $w$ is an at most exponentially long word (given on a special input tape);
\item based on this method, approximate $\pi_i(\tW_i)$ in PSPACE.
\end{enumerate}
For item~1 we use floating-point arithmetic, for item~2 we use Ladner's result~\cite{Ladner89} on counting in polynomial space.

For $k \in \N$, define $\F_k := \{m \cdot 2^z : z \in \Z,\ 0 \le m \le 2^k-1\}$, the set of \emph{$k$-bit floating-point numbers}.
For our purposes, nonnegative floating-point numbers suffice, and there is no need to bound the exponent~$z$, as all occurring exponents will have polynomial bit size.
We define rounding as usual: for $x \ge 0$ write $\round{x}_k$ for the number in~$\F_k$ that is nearest to~$x$ (break ties in an arbitrary but deterministic way).
Then there is $\delta$ with $\round{x}_k = x \cdot (1+\delta)$ and $|\delta| < 2^{-k}$, see \cite[Theorem~2.2]{Higham}.
A standard analysis of rounding errors in finite-precision arithmetic~\cite[Chapter~3]{Higham} yields the following lemma:
\begin{lemma} \label{lem-error-analysis}
Let $\pi$ be an initial distribution and $0 < \theta < 1$.
Let $k \in \N$ be such that $2^k \ge 2 (n+1) |Q| / \theta$.
Let $w = a_1 a_2 \cdots a_m \in \Sigma^*$ with $m \le n$.
Compute $\tpi(w)$ as
\[ ((\cdots((\pi \cdot M(a_1)) \cdot M(a_2)) \cdots ) \cdot M(a_{m})) \cdot \eta^\top \,,
\]
where rounding~$\mathord{\round{\cdot}_k}$ is applied after each individual (scalar) multiplication and addition.
Then $\tpi(w) \in [\pi(w) (1 - \theta), \pi(w) (1 + \theta)]$.
\end{lemma}
\begin{proof}
For all $i \in \N$ write $\gamma_i := i \cdot 2^{-k} / (1 - i \cdot 2^{-k})$.
By \cite[Equation~(3.11)]{Higham} there are matrices $\Delta_1, \ldots, \Delta_m$ and a vector $\widetilde{\eta}$ such that
\[
 \tpi(w) \ = \ \pi \cdot (M(a_1) + \Delta_1) \cdot (M(a_2) + \Delta_2) \cdots (M(a_m) + \Delta_m) \cdot (\eta + \widetilde{\eta})^\top
\]
and $|\Delta_i| \le \gamma_{|Q|} M(a_i)$ and $|\widetilde{\eta}| \le \gamma_{|Q|} \eta$, where by $|\Delta_i|$ and $|\widetilde{\eta}|$ we mean the matrix and vector obtained by taking the absolute value componentwise.
(In words, the result $\tpi(w)$ of the floating-point computation is the result of applying an exact computation with slightly perturbed data---a ``backward error'' result.)
It follows:
\begin{align*}
|\tpi(w) - \pi(w)| 
\ &\le \ \bigg( -1 + \prod_{j=1}^{m+1}\left(1 + \gamma_{|Q|}\right)\bigg) \pi(w) && \text{by \cite[Lemma~3.8]{Higham}} \\
\ &\le \ \gamma_{(m+1) \cdot |Q|} \pi(w) && \text{by \cite[Lemma~3.3]{Higham}} \\
\ &\le \ 2 (n+1) |Q| \cdot 2^{-k} \pi(w) && \text{as } (n+1) |Q| \cdot 2^{-k} \le 1/2 \\
\ &\le \ \theta \pi(w) && \hfill \qedhere
\end{align*}
\end{proof}
%
The development so far suggests the following approximation approach:
Let $\varepsilon>0$ be the error bound from the input. 
Let $n \in \N$ be the number from Lemma~\ref{lem-approximation-length-bound}, where $\lambda$ is such that $2 \lambda = \varepsilon/2$.
Let $k \in \N$ be the smallest number such that $2^k \ge 2 (n+1) |Q| / \theta$, where $\theta$ is such that $4 \theta = \varepsilon/2$.
Observe that $k$ (the bit size of~$2^k$) is polynomial in the input.
Define, for each word~$w$ and both~$i$, the approximation~$\tpi_i(w)$ as in Lemma~\ref{lem-error-analysis}.
This defines also $\tW_1, \tW_2$.
By~\eqref{eq-approx-theta-epsilon} we have:
\begin{equation*}
\pi_1(\tW_1) + \pi_2(\tW_2) - \frac{\varepsilon}{2} \quad \le \quad  1 - d(\pi_1, \pi_2) \quad \le \quad \pi_1(\tW_1) + \pi_2(\tW_2) + \frac{\varepsilon}{2}
\end{equation*}
Thus we can complete the proof of Theorem~\ref{thm-approximation-PSPACE} by proving the following lemma:
\begin{restatable}{lemmaS}{lemapproxtWi} \label{lem-approx-tWi}
For both~$i$, one can approximate $\pi_i(\tW_i)$ within $\varepsilon/4$ in PSPACE.
\end{restatable}
\begin{proof}
We discuss only the approximation of~$\pi_1(\tW_1)$; the case of~$\pi_2(\tW_2)$ is similar.

Construct a ``probabilistic PSPACE Turing machine''~$\T$ that samples a random word~$w$ according to~$\Pr_{\pi_1}$.
For that, $\T$~uses probabilistic branching according to the transition probabilities in~$M$.
While producing~$w$ in this way, but without storing~$w$ as a whole, $\T$~computes also the values $\tpi_1(w), \tpi_2(w)$ according to Lemma~\ref{lem-error-analysis}.
If and when $w$ gets longer than~$n$ then $\T$~rejects.
If $\tpi_1(w) < \tpi_2(w)$ then $\T$~accepts; otherwise $\T$~rejects.
The probability that $\T$~accepts equals $\pi_1(\tW_1)$.
This probability can be computed in PSPACE by Ladner's result~\cite{Ladner89} on counting in polynomial space.
To be precise, note that this probability is a fraction $p/q$ of two natural numbers $p,q$ of at most exponential bit size.
By Ladner's result one can compute arbitrary bits of $p,q$ in PSPACE.
Hence an approximation within~$\varepsilon/4$ can also be computed in PSPACE.
Technical details about how we apply Ladner's result are provided \iftechrep{in Appendix~\ref{app-ladner}}{in~\cite{18K:ICALP-TR}}.
\end{proof}

\section{Open Problems}

In this paper we have considered the total variation distance between the distributions on finite words that are generated by two LMCs.
In a more general version of LMCs, the end-of-word probabilities are zero, so that the LMC generates infinite words.
The production of finite words $w \in \Sigma^*$ can be simulated by producing $w \eow \eow \eow \cdots$ where $\eow$~is an end-of-word symbol.
It follows that the undecidability and hardness results of this paper apply equally to infinite-word LMCs.
In fact, all these results strengthen those from~\cite{14CK-LICS}, where the total variation distance between infinite-word LMCs is studied.
The PSPACE approximation algorithm in this paper (Theorem~\ref{thm-approximation-PSPACE}) applies only to finite words, and the author does not know if it can be generalized to infinite-word LMCs.
Whether the non-strict threshold-distance problem is decidable is open, both for finite- and for infinite-word LMCs.

Another direction concerns LMCs that are \emph{not hidden}, i.e., where each emitted label identifies the next state; or, slightly more general, \emph{deterministic} LMCs, i.e., where each state and each emitted label identify the next state.
The reduction that shows square-root-sum hardness in~\cite[Theorem 15]{14CK-LICS} also applies to the threshold-distance problem for deterministic finite-word LMCs, but the author does not know a hardness result for approximating the distance between deterministic LMCs.

\subparagraph*{Acknowledgement.}
The author thanks anonymous referees for their helpful comments.

\bibliography{db}

\iftechrep{

\newpage
\appendix
\section{Missing Proofs}

\subsection{Proof of Proposition~\ref{prop-threshold-distance-acyclic}}
Here is Proposition~\ref{prop-threshold-distance-acyclic} from the main text:
\propthresholddistanceacyclic*
\begin{proof}
Since the problem \#NFA is \#P-complete~\cite{KSM-counting-strings}, the problem whether $|L(\A) \cap \Sigma^n| \le u$ holds (for given NFA~$\A$ and given $n \in \N$ in unary and given $u \in \N$ in binary) is PP-hard.
Indeed, reduce, using \#P-hardness of \#NFA, the canonical PP-complete problem, MajSAT, to the problem whether $|L(\A) \cap \Sigma^n| > u$ holds.
It follows that the problem whether $|L(\A) \cap \Sigma^n| \le u$ holds is PP-hard, as PP is closed under complement.
(The problem is in PP, hence PP-complete, but this is not needed.)

We may assume $u \le |\Sigma^n|$ as otherwise the problem is trivial.
Using the construction from Lemma~\ref{lem-non-strict-threshold-distance} we have:
\begin{alignat*}{2}
  && \ |L(\A) \cap \Sigma^n| &\ \le \ u \\
\Longleftrightarrow\qquad && \ |\Sigma^n \setminus L(\A)| &\ \ge \ |\Sigma|^n - u \\
\Longleftrightarrow\qquad && \ d(\pi_1, \pi_2) &\ \ge \ y + \frac{|\Sigma|^n - u}{|\Sigma|^n |Q|^n}
\intertext{%
Deciding the last inequality is an instance of the non-strict threshold-distance problem.
Concerning the strict threshold-distance problem, replace $|L(\A) \cap \Sigma^n| \le u$ by $|L(\A) \cap \Sigma^n| < u+1$, which leads to
}
 && d(\pi_1, \pi_2) &\ > \ y + \frac{|\Sigma|^n - (u+1)}{|\Sigma|^n |Q|^n}\,,
\end{alignat*}
an instance of the strict threshold-distance problem.
\end{proof}

\subsection{Proofs of Theorem~\ref{thm-threshold-distance-acyclic} and Remark~\ref{rem-dk}} \label{app-thm-threshold-distance-acyclic}
Here is Theorem~\ref{thm-threshold-distance-acyclic} from the main text:
\thmthresholddistanceacyclic*
\begin{proof}
In view of Proposition~\ref{prop-threshold-distance-acyclic} it suffices to prove membership in PP.
We start with the strict threshold-distance problem.
Following the definition of~PP, it suffices to construct a nondeterministic polynomial-time bounded Turing machine~$\T$ with the following property: if the input of~$\T$ is an acyclic LMC $\M = (Q, \Sigma, M, \eta)$ and initial distributions $\pi_1, \pi_2$ and a threshold~$\tau$ then $\T$ accepts on at least half of its computations if and only if $d(\pi_1, \pi_2) > \tau$.
In the following we describe the operation of~$\T$.

Let the input of~$\T$ be $\M, \pi_1, \pi_2, \tau$, as above.
The Turing machine~$\T$ computes the product, $D \in \N$, of all denominators that appear in the encodings of $\pi_1, \pi_2, M, \eta, \tau$.
Then $\T$ computes integer vectors $\pi_1', \pi_2', \eta'$ and integer matrices~$M'(a)$ for all $a \in \Sigma$ such that $\pi_1 = \pi_1'/D$ and $\pi_2 = \pi_2'/D$ and $\eta' = \eta/D$ and $M(a) = M'(a)/D$ for all $a \in \Sigma$.
For all $w \in \Sigma^*$ we have $M(w) = M'(w)/D^{|w|}$ and thus for $i \in \{1,2\}$:
\begin{equation}
 \pi_i(w) = \frac{\pi_i' M'(w) (\eta')^\top}{D^{|w|+2}} \label{eq-PP-upper-pii}
\end{equation}
Then $\T$ computes (using a polynomial-time graph algorithm on the transition graph) the length, $n$, of a longest word~$w$ with $\pi_1(w) > 0$ or $\pi_2(w) > 0$.
We have:
\begin{align*}
d(\pi_1, \pi_2) \ 
= & \ \frac12 \sum_{w \in \Sigma^*} | \pi_1(w) - \pi_2(w) | && \text{Proposition~\ref{prop-tv-dist-basic}} \\
= & \ \frac12 \sum_{w \in \Sigma^{\le n}} | \pi_1(w) - \pi_2(w) | && \text{definition of~$n$} \\
= & \ \frac12 \sum_{w \in \Sigma^{\le n}} \frac{D^{n - |w|} \left|(\pi_1' - \pi_2') M'(w) (\eta')^\top \right|}{D^{n+2}} && \text{by~\eqref{eq-PP-upper-pii}}
\end{align*}
It follows:
\begin{align*}
                    & \ \ d(\pi_1, \pi_2) \ > \ \tau \\
\Longleftrightarrow & \ \ \sum_{w \in \Sigma^{\le n}} D^{n - |w|} \left|(\pi_1' - \pi_2') M'(w) (\eta')^\top\right| \ > \ 2 D^{n+2} \tau \\
\Longleftrightarrow & \ \ \sum_{w \in \Sigma^{\le n}} D^{n - |w|} \left|(\pi_1' - \pi_2') M'(w) (\eta')^\top\right| \ \ge \ 2 D^{n+2} \tau + 1
\end{align*}
The last equivalence holds as both sides of the inequality are integers.

So far $\T$ has operated deterministically and in polynomial time.
Now $\T$~branches nondeterministically.
It either enters one of $2 D^{n+2} \tau + 1$ rejecting computations or guesses a word $w \in \Sigma^{\le n}$.
In the latter case $\T$ then computes $a := D^{n - |w|} \left|(\pi_1' - \pi_2') M'(w) (\eta')^\top\right|$ and branches nondeterministically in one of~$a$ accepting computations.
Thus, by the computation above, $\T$ accepts on at least half of its computations if and only if $d(\pi_1, \pi_2) > \tau$.

The modification required for the non-strict threshold-distance problem is straightforward: replace $2 D^{n+2} \tau + 1$ with $2 D^{n+2} \tau$.
\end{proof}

In Remark~\ref{rem-dk} we mention that the PP upper bound generalizes to all~$d_k$.
Indeed, we have:
\[
| \pi_1(w) - \pi_2(w) |^k
\ = \
\frac{D^{(n - |w|)k} \left|(\pi_1' - \pi_2') M'(w) (\eta')^\top \right|^k}{D^{(|w|+2)k}}
\]
The other modifications of the PP-membership argument in the proof of Theorem~\ref{thm-threshold-distance-acyclic} are straightforward.
\qed

\subsection{Proof of Lemma~\ref{lem-approximation-length-bound}}
Here is Lemma~\ref{lem-approximation-length-bound} from the main text:
\lemapproximationlengthbound*
\begin{proof}
Let $\pmin > 0$ be the smallest nonzero probability that appears in the description of~$\M$.
If $\pmin = 1$ then $\pi_i(w) = 0$ for all $w$ with $|w| \ge |Q|$.
So we can assume $\pmin < 1$ for the rest of the proof.
We also assume $0 < \lambda \le 1$.

For any state~$q$, consider the shortest path to a state~$q'$ with $\eta(q') > 0$.
The probability of taking this path (of length at most $|Q|-1$, hence emitting at most $|Q|-1$ labels) and then ending the word is at least~$\pmin^{|Q|}$, i.e., for all~$q$ we have
$\delta_q(\Sigma^{<|Q|}) \ge \pmin^{|Q|}$, equivalently, $\delta_q(\Sigma^{\ge |Q|}) \le 1-\pmin^{|Q|}$.
By repeating this argument it follows that for all $k \in \N$ we have
\begin{alignat}{2} \delta_q\left(\Sigma^{\ge k |Q|}\right) & \ \le \ \left(1-\pmin^{|Q|}\right)^k && \qquad\text{for all } q \in Q\,, \notag
\intertext{hence}
\pi_i\left(\Sigma^{\ge k |Q|}\right) & \ \le \ \left(1-\pmin^{|Q|}\right)^k && \qquad\text{for } i \in \{1,2\}\,.
\label{eq-approximation-length-bound-1}
\end{alignat}
Choose now $k \in \N$ such that $k \ge (-\ln \lambda) / \pmin^{|Q|}$.
Such $k$ can be computed in polynomial time.
Since $x \le -\ln(1-x)$ for $x < 1$, we have
\[
 k \ \ge \ \frac{\ln \lambda}{\ln \left(1-\pmin^{|Q|}\right)}
\qquad \text{and hence} \qquad
\left(1-\pmin^{|Q|}\right)^k \ \le \ \lambda\,.
\]
By~\eqref{eq-approximation-length-bound-1} it follows that $\pi_i\left(\Sigma^{\ge k |Q|}\right) \le \lambda$ holds for both $i \in \{1,2\}$.
\end{proof}

\subsection{Technical Details of the Proof of Lemma~\ref{lem-approx-tWi}} \label{app-ladner}
Here is Lemma~\ref{lem-approx-tWi} from the main text:
\lemapproxtWi*
\begin{proof}
We provide technical details for the proof given in the main text.

In more detail, Ladner defines the function class \#PSPACE in analogy to \#P and shows~\cite[Theorem~1]{Ladner89} that \#PSPACE = FPSPACE where FPSPACE is the class of functions with values in~$\N$ that are computable in PSPACE.
We can build a nondeterministic PSPACE Turing machine~$\T'$ that simulates the probabilistic branching of~$\T$ (from the main text) by nondeterministic branching:
The Turing machine~$\T'$ computes the product, $D \in \N$, of all denominators that appear in the encodings of $\pi_1, \pi_2, M, \eta$.
Whenever $\T$~would branch probabilistically, say with probability~$x$, the Turing machine $\T'$ branches nondeterministically in $x \cdot D$ computation branches;
note that $x \cdot D \in \N$.
If the sampled word~$w$ is shorter than~$n$, then $n-|w|$ dummy labels are sampled so that $\T'$~has $D^{n+2}$ computations in total (taking $\pi_1$ and $\eta$ into account).
Since \#PSPACE = FPSPACE, one can compute the number of accepting computations in PSPACE.
This number divided by the total number, $D^{n+2}$, of computations equals the acceptance probability of~$\T$, which equals $\pi_1(\tW_1)$.
\end{proof}

}{}

\end{document}